\documentclass[a4paper]{article}
\usepackage[margin=3cm]{geometry}
\usepackage{latexsym}
\usepackage{amsmath}
\usepackage{amsthm}
\usepackage{amssymb}
\usepackage{amsfonts}
\usepackage[pdftex]{graphicx}

\newtheorem{theorem}{Theorem}
\newtheorem{remark}{Remark}
\newtheorem{definition}[theorem]{Definition}

\newtheorem{lemma}{Lemma}
\newtheorem{corollary}{Corollary}

\begin{document}

\title{Continuous approximations of a class of piece-wise continuous systems}

\author{MARIUS-F. DANCA\\Department of Mathematics and Computer Science, Avram Iancu University, \\Str. Ilie Macelaru, nr. 1A, 400380 Cluj-Napoca, Romania,\\and\\Romanian Institute of Science and Technology, \\Str. Ciresilor nr. 29, 400487 Cluj-Napoca, Romania}

\maketitle

\begin{abstract}
In this paper we provide a rigorous mathematical foundation for continuous approximations of a class of systems with piece-wise continuous functions. By using techniques from the theory of differential inclusions, the underlying piece-wise functions can be locally or globally approximated. The approximation results can be used to model piece-wise continuous-time dynamical systems of integer or fractional-order. In this way, by overcoming the lack of numerical methods for differential equations of fractional-order with discontinuous right-hand side, unattainable procedures for systems modeled by this kind of equations, such as chaos control, synchronization, anticontrol and many others, can be easily implemented. Several examples are presented and three comparative applications are studied.
\end{abstract}


\section{Introduction}

Despite the doubths in the 17th century regarding the practical applicability of fractional derivatives since fractional derivatives have no clear geometrical interpretations \cite{pod}, there are nowadays a lot of works on systems of fractional-order and their related applications in many domains, such as physics, engineering, mathematics, finance, chemistry, and so on (see, for example, the books \cite{old,petr} or the papers of Caputo \cite{cap}).

On the other hand, discontinuous functions can be found in two-dimensional mechanical systems such as systems with dry friction, oscillating systems with combined dry and viscous damping, forced vibrations, brake processes with locking phases, control synthesis of uncertain systems, control theory, calculus of variations, systems with stick and slip modes, braking processes with locking phases, PDEs, elastoplasticity, and also in  game theory, optimization, biological and physiological systems, electrical (chaotic) circuits, networks, power electronics etc (see e.g. \cite{wir,cort}, \cite{mario}, and the references therein).

\noindent Therefore, dynamical systems of fractional-order, modeled with piece-wise continuous functions, gain more and more interest for real systems which follow behaviors modeled better with fractional-order equations than of integer order.

Although there are numerical methods for fractional-order DE (see e.g. \cite{kai,kai2,dr}) and also for DE with discontinuous right-hand side (see e.g. \cite{don,bro,lem,mares}), to the best of our knowledge, there are no numerical methods for DE of fractional-order with discontinuous right-hand side. Consequently, modeling continuously or smoothly the underlying systems could be of a real interest for example in chaos control, synchronization, anticontrol and so on, and also for quantitative analysis.

The class of piece-wise continuous functions $f:\mathbb{R}^n\rightarrow \mathbb{R}^n$ defining these systems, which will be continuously approximated, has the following form:
\begin{equation}\label{f}
f(x(t))=g(x(t))+A(x(t))s(x(t)),
\end{equation}

\noindent with $g:\mathbb{R}^n\rightarrow \mathbb{R}^n$ a single-valued, vector, at least continuous function, and $s:\mathbb{R}^n\rightarrow \mathbb{R}^n$, $s(x)=(s_1(x_1),s_2(x_2),...,s_n(x_n))^T$ a vector-valued piece-wise function, with $s_i:\mathbb{R}\rightarrow \mathbb{R}$, $i=1,2,...,n$, real piece-wise constant functions, $A_{n\times n}$ a square matrix of real functions.

The following assumption will be needed:

\vspace{3mm}
\noindent (\textbf{H1}) The function $As$ is discontinuous in at least one of its components
\vspace{3mm}

The smoothness of the functions $x^msgn(x)$, for different values of $m$ has been discussed in \cite{dan1}.

This form of $f$, appearing in the great majority of nonlinear piece-wise continuous systems of fractional or integer order, is modeled by the following Initial Value Problem (IVP):
\begin{equation}
\label{IVP0}
D_*^q x(t)=f(x(t)),~~
x(0)=x_0,~~~ t\in I=[0,\infty).
\end{equation}

\noindent Here, with $D_*^q$, $0<q\leq1$ ($q=1$ for the integer order), denotes the operator commonly used in fractional calculus: Caputo's differential operator of order $q$ (called also \emph{smooth fractional derivative} with starting point 0)
\begin{equation*}
D_*^qx(t)=\frac{1}{\Gamma(1-q)}\int_0^t (t-\tau)^{-q}\frac {d}{dt}x(\tau)d\tau,
\end{equation*}

\noindent with $\Gamma$ Euler's Gamma function
\begin{equation*}
\Gamma(z)=\int_0^t t^{z-1}e^{-t}dt, ~~~ z\in\mathbb{C}, ~~Re(z)>0.
\end{equation*}

To overcome the problem of numerical integration of systems modeled by (\ref{IVP0}), the discontinuous problem will be transformed into a continuous one. For this purpose, Filippov's approach \cite{fil} will be used along with some basic results from the theory of fractional differential inclusions \cite{aub1,aub2}.

The obtained approximation results, targeting the piece-wise constant functions $s$, are valid for a large class of functions, such as \emph{Heaviside} function $H$, \emph{rectangular} function (as difference of two Heaviside functions), and \emph{signum}, one of the mostly encountered in practical applications.

The null set of the discontinuity points of $f$, $\mathcal{M}$ (with zero Lebesgue measure $\mu$), is generated by the discontinuity points of the components $s_i$.

Because the systems modeled by the IVP (\ref{IVP0}) are autonomous, hereafter, unless otherwise mentioned, we drop the time variable in writing.

Next, consider some examples.

The piece-wise linear one-dimensional function $f:\mathbb{R}\rightarrow \mathbb{R}$
\begin{equation*}
f(x)=2-3sgn(x),
\end{equation*}

\noindent has $\mathcal{M}=\{0\}$ and the graph shown in Fig. \ref{fig0}.

Generally, discontinuous dynamical systems (of integer or fractional-order) can be found in $\mathbb{R}^2$, such as the following fractional variant of the system which models a unit mass which is subject of a discontinuous spring force \cite{cort}:
\begin{equation}
\begin{array}{lc}
D_*^{q_1}{x}_{1}=x_{2}, &  \\
D_*^{q_2}{x}_{2}=-sgn(x_{1}), &
\end{array}%
\end{equation}

\noindent where
\begin{equation*}
g(x)=\left(
\begin{array}{c}
x_{2} \\
0%
\end{array}%
\right) ,~~~A=\left(
\begin{array}{cc}
0 & 0 \\
-1 & 0%
\end{array}%
\right) ,~\ ~s(x)=\left(
\begin{array}{c}
sgn(x_{1}) \\
sgn(x_{2})%
\end{array}%
\right).
\end{equation*}

\noindent In this case, $\mathcal{M}=\{(0,x_2), x_2\in \mathbb{R}\}$.

Another typical example of a mechanical system, which models a friction oscillator \cite{wir} and can model wings of insects \cite{mark}, is governed by:
\begin{equation*}
\overset{..}{x}+\lambda \overset{.}{x}+x^{3}+\varphi (x,\overset{.}{x})sgn(%
\overset{.}{x})=0,
\end{equation*}

\noindent where $\varphi$ is some function, $\lambda$ is the bifurcation parameter. After the system is written in the standard form, we have:
\begin{equation*}
g(x)=\left(
\begin{array}{c}
x_{2} \\
-\lambda x_{1}-x_{1}^{3}%
\end{array}%
\right) ,~~~A=\left(
\begin{array}{cc}
0 & 0 \\
0 & -\varphi \left( x_{1},x_{2}\right)
\end{array}%
\right).
\end{equation*}

The following three-dimensional system is a fractional variant of the discontinuous Chua system \cite{chua}:
\begin{equation}\label{chu}
\begin{array}{cl}
D_*^{q_1}(x)= & -2.571x_{1}+9x_{2}+3.857sgn(x_{1}), \\
D_*^{q_2}(x)= & x_{1}-x_{2}+x_{3}, \\
D_*^{q_3}(x)= & -px_{3},
\end{array}%
\end{equation}

\noindent where $p\in \mathbb{R}$ is the bifurcation parameter (see the graph of the piece-wise continuous component in Fig. \ref{zozo}(a)).

This paper is organized as follows: In Section \ref{switch}, the approximation of $f$ defined by (\ref{f}) is presented, while in Section 3, three applications are analyzed.

\section{Approximation of $f$} \label{switch}

In this section we show how piece-wise continuous functions $f$ modeled by (\ref{f}) can be continuously approximated. Precisely, since $g$ is continuous, we are interested in approximating the piece-wise-constant functions $s_i$.

For this purpose, let us consider the IVP (\ref{IVP0}) whose right-hand side will be first transformed into a set-valued function via the Filippov regularization \cite{fil}. In this way, the single-valued initial problem is reformulated as a set-valued one, namely a differential inclusion of fractional-order of the form
\begin{equation}
D_*^q{x}\in F(x), ~~x(0)=x_0,~~\text{for} ~a.a. ~t\in I,
\end{equation}

\noindent where $F:\mathbb{R}^n \rightrightarrows \mathbb{R}^n$ is a set-valued vector function, mapping into the set of subsets of $\mathbb{R}^n$, which can be defined in several ways (see \cite{jon}, one of a few related works on fractional differential inclusions).

A simple (convex) expression of a set-valued function $F$ is obtained by the so-called Filippov regularization \cite{fil,aub1,aub2}:
\begin{equation}\label{fil}
F(x)=\bigcap_{\varepsilon >0}\bigcap_{\mu(\mathcal{M})=0} \overline{conv}(f({z\in \mathbb{R}^n: |z-x|\leq\varepsilon}\backslash \mathcal{M})).
\end{equation}

As can be seen, $F(x)$ is the convex hull of $f(x)$ (see the sketch in Fig. \ref{fig00}, parts (a), (b)), $\mu$ being the Lebesgue measure and $\varepsilon$ the radius of the ball centered at $x$. At those points where $f$ is continuous, $F(x)$ consists of one single point, which coincides with the value of $f$ at this point (i.e. we get $f(x)$ back as the right-hand side: $F(x)=\{f(x)\}$). At the points belonging to $M$, $F(x)$ is given by (\ref{fil}) (Fig. \ref{fig00}(c)).

More on the Filippov regularization and generalized solutions to discontinuous equations can be found in, e.g. the review papers \cite{cort}, \cite{haj}.

In order to justify the use of the Filippov regularization to some physical systems, we must choose small values for $\varepsilon$, so that the motion of the physical systems is arbitrarily close to a certain solution of the underlying differential inclusion (it tends to the solution, as $\varepsilon\rightarrow 0$).

If the piece-wise-constant functions $s_i$ are $sgn$, their set-valued form, obtained with Filippov regularization and denoted usually by $Sgn:\mathbb{R}\rightrightarrows \mathbb{R}$, is defined as follows (see Fig. \ref{fig1}(a) before regularization and Fig. \ref{fig1}(b) after regularization):
\begin{equation}
Sgn(x)=\left\{
\begin{array}{cc}
\{-1\} & x<0, \\
\lbrack -1,1] & x=0, \\
\{+1\} & x>0.%
\end{array}%
\right.
\end{equation}

By applying the Filippov regularization to $f$, one obtains the following set-valued function
\begin{equation}
\label{IVP1}
F(x):=g(x)+A(x)S(x),
\end{equation}

\noindent with
\begin{equation}\label{s}
S(x)=(S_1(x_1),S_2(x_2),...,S_n(x_n))^T,
\end{equation}

\noindent where $S_i:\mathbb{R}\rightarrow \mathbb{R}$ is the set-valued variant of $s_i$, $i=1,2,...,n$ ($Sgn(x_i)$ in general).

Because the set-valued character of $F$ in (\ref{IVP1}) is generated by $S_i$, which are real functions, the notions and results presented next are considered in $\mathbb{R}$, for the case of $n=1$, but they are also valid in the general cases of $n>1$.

Next, consider some general set-valued function $F:\mathbb{R}\rightrightarrows \mathbb{R}$.

The graph of a set-valued function $F$ is defined as follows:
\begin{equation*}
Graph(F):=\{(x,y)\in \mathbb{R}\times \mathbb{R}, ~y\in F(x)\}.
\end{equation*}

\begin{remark}\label{rem}Due to the symmetric interpretation of a set-valued function as a graph (see e.g. \cite{aub1}), we say that a set-valued function satisfies a property if and only if its graph satisfies it. For instance, a set-valued function is said to be closed if and only if its graph is closed.
\end{remark}

\begin{definition}
As set-valued function $F$ is upper semicontinous (u.s.c.) at $x^0\in \mathbb{R}$, if for any open set $B$ containing $F(x^0)$, there exists a neighborhood $A$ of $x^0$ such that $F(A)\in B$.
\end{definition}

\noindent We say that $F$ is u.s.c. if it is so at every $x^0\in \mathbb{R}$.

\vspace{3mm}
\noindent U.s.c., which is a basic property, practically means that the graph of $F$ is closed.

\begin{definition}
A single-valued function $h:\mathbb{R}\rightarrow \mathbb{R}$ is called an \emph{approximation} (\emph{selection}) of the set-valued function $F$ if
\begin{equation*}
h(x)\in F(x),~~\forall	x\in \mathbb{R}.
\end{equation*}
\end{definition}

\noindent Generally, a set-valued function admits (infinitely) many approximations (see Fig. \ref{fig1}(c) for the case of $Sgn$ function). For the theory of selections for set-valued functions, compare \cite{aub1,aub2} and \cite{kas}.

\smallskip
\noindent \textbf{Notation}

\noindent Let $\mathcal{C}_\varepsilon^0(\mathbb{R})$ be the class of real continuous approximations $\widetilde{s}:\mathbb{R}\rightarrow\mathbb{R}$ of the set-valued function $F$ which satisfy

(i) $Graph(\widetilde{s})\subset Graph (B(F,\varepsilon))$.

(ii) For every $x\in \mathbb{R}$, $\widetilde{s}(x)$ belongs to the convex hull of the image of $F$.

\smallskip

\noindent Above, $B(x,\varepsilon)$ is the disk of radius $\varepsilon$ centered at $x$.

The set-valued functions $S_i$, $i=1,2,...,n$, can be approximated due to the Approximate Theorem, called also Cellina's Theorem (see \cite{aub1} p. 84 and \cite{aub2} p. 358), which states that a set-valued function $F$, with closed graph and convex values, admits $\mathcal{C}_\varepsilon^0$ approximations.

\begin{remark}
Cellina's Theorem provides locally Lipschitzean approximations. Since locally Lipschitzean  functions are also continuous, in this paper we will consider $C_\varepsilon^0$ approximations.
\end{remark}

\subsection{Global approximation}\label{global}

The global approximation of $S_i$, $i=1,2,...,n$, defined over $\mathbb{R}$, is assured by the following lemma.

\begin{lemma}\label{tprinc}
For every $\varepsilon>0$, the set-valued functions $S_i$, $i=1,2,...,n$, admit global $\mathcal{C}_\varepsilon^0$ approximations.
\end{lemma}
\begin{proof} $S_i$, for $i=1,2,...,n$, are convex u.s.c. (see e.g. the Remark in \cite{fil} p. 43 or Example in \cite{aub2} p. 39 for u.s.c. property) and, via Remark \ref{rem}, are non-empty closed valued functions. Therefore, they verify Cellina's Theorem which guaranties the existence of $\mathcal{C}_\varepsilon^0$ approximations on $\mathbb{R}$.
\end{proof}
\vspace{3mm}
\textbf{Notation} Denote by $\widetilde{s}_i:\mathbb{R}\rightarrow \mathbb{R}$ the global approximations of $S_i$.
\vspace{3mm}

For the sake of simplicity hereafter, $\varepsilon$ is considered as having the same value for each component $\widetilde{s}_i(x_i)$, $i=1,2,...,n$. Also, the index $i$ will be dropped, unless specified.

The constructive proof of Cellina's Theorem allows us to ease the approximations choice. Any single-valued function on $\mathbb{R}$, with the graph in the $\varepsilon$-neighborhood, is an approximate selection of $S$ from the Celina Theorem. However, some of the best candidates for $\widetilde{s}$ are the \emph{sigmoid} functions which provide the required flexibility and to which the abruptness of the discontinuity can be easily modified. If $S(x)=Sgn(x)$, one of the mostly utilized sigmoid approximations is the following function $\widetilde{sgn}$\footnote{Sigmoid functions include the ordinary arctangent such as $\frac{2}{\pi}arctan\frac{x}{\varepsilon}$, the hyperbolic tangent, the error function, the logistic function, algebraic functions like $\frac{x}{\sqrt{\epsilon+x^2}}$, and so on.}:
\begin{equation}\label{h_simplu}
\widetilde{sgn}(x)=\frac{2}{1+e^{-\frac{x}{\delta}}}-1\approx Sgn(x),
\end{equation}

\noindent where $\delta$ is a positive parameter which controls the slope in the neighborhood of the discontinuity manifold $x=0$. In Fig. \ref{fig2}(a), $\tilde{sgn}$ is plotted as a function of $\delta$, while in Fig. \ref{fig2} b, it is plotted for two distinct values.

The smallest $\varepsilon$ values, necessarily to embed $\widetilde{sgn}$ within an $\varepsilon$-neighborhood of $Sgn$ (as stated by Cellina's Theorem), depends proportionally on $\delta$. However, finding an explicit relation for $\delta$ as a function of $\varepsilon$ is a difficult task. Moreover, for $x\neq 0$, $\widetilde{sgn}$ is identical to the single-valued branches of $Sgn$ (the horizontal lines $\pm1$) only asymptotically, for $x\rightarrow \pm\infty$. For example, for $\delta=1/100$, at the point $x=0.06$, the difference is of order of $10^{-3}$, even the two graphs look apparently identical in the underlying points $A$ or $B$ (Fig. \ref{fig2}(c)). To reduce the size of $\varepsilon$ to e.g. $10^{-4}$, $\delta$ should be of order of $10^{-5}$. From theoretical point of view, the approximation can be obtained with infinity precision.

For the Heaviside function, which in its piece-wise constant variant can be expressed in terms of the $signum$ function by $H(x)=\frac{1}{2}[1+sgn(x)]$, the approximate sigmoid function (\ref{h_simplu}) becomes $\widetilde{H}(x)=\frac{1}{1+e^{-\frac{x}{\delta}}}$.

Now, we can derive the following result, which assures the possibility to approximate $f$ globally.

\begin{theorem}\label{th1}
Let $f$ be defined by (\ref{f}). If $g$ is continuous, then for every $\varepsilon>0$, there exist global approximations of $f$, $\tilde{f}:\mathbb{R}^n\rightarrow \mathbb{R}^n$, such that
\begin{equation}\label{glo}
\tilde{f}(x)=g(x)+A(x)\widetilde{s}(x)\approx f(x),
\end{equation}

\end{theorem}

\noindent Theorem \ref{th1} states that systems modeled by the IVP (\ref{IVP0}) can be continuously modeled by the following IVP:
\begin{equation*}
D_*^q(x)=\tilde{f}(x),
\end{equation*}

\noindent with $\tilde{f}$ defined by (\ref{glo}).

\noindent For example, the function
\begin{equation}\label{fifi}
f(x)=-x^2+sgn(x-0.5)
\end{equation}

\noindent can be globally continuously approximated on $\mathbb{R}$ (dotted line in Fig. \ref{fig3}), having the approximated form
\begin{equation}
\tilde{f}(x)=-x^2+\widetilde{sgn}(x-0.5)=-x^2+\frac{2}{1+e^{-\frac{x-0.5}{\delta}}}-1.
\end{equation}

 \noindent Here, $f$ is transformed first into the set-valued function $F(x)=-x^2+Sgn(x-0.5)$ (red line in Fig. \ref{fig3}), and then approximated via $\widetilde{sgn}$.

\subsection{Local approximation}

In order not to affect significantly the physical characteristics of the underlying system, it is desirable to approximate $S$ only on some tight $\varepsilon$-neighborhoods of the discontinuity $x=0$, not on the entire real axis, as global approximations given by Theorem \ref{tprinc} do, when the difference error between $S$ and $\tilde s$ persists along the entire real axis $\mathbb{R}$. Graphically speaking, we want to restrict the approximation of $S$ only within a narrow vertical band of width $\varepsilon$ centered along the vertical axis. This is allowed by the particular (convex u.s.c.) form of the set-valued functions $S$, and by the great flexibility afforded by continuous functions which can be glued or pasted, without altering the continuity property. In this case, the global approximations (such as the sigmoid functions) are no longer useful (see Fig. \ref{fig2}(b) for the case of $Sgn$), but other kind of selections, such as the local approximations, can be used.

The following corollary, which is a simple consequence of the results of Subsection \ref{global}, ensures the existence of local approximations for the set-valued functions $S$ (see the sketch in Fig. \ref{fig6}).

\begin{corollary}\label{coro2}
For every $\varepsilon>0$, $S$ admits locally $\mathcal{C}_\varepsilon^0$ approximations  $\widetilde{s}_{\varepsilon }:(-\varepsilon,\varepsilon)\rightarrow \mathbb{R}$, which verify the neighborhood continuity conditions
\begin{equation}\label{cond}
\widetilde{s}_{\varepsilon}(\pm\varepsilon)=S(\pm\varepsilon).
\end{equation}
\end{corollary}

\noindent In this case, $\tilde{s}_{\varepsilon }$ can also be continuously extended on $\mathbb{R}$, obtaining a new global approximation $\tilde{\tilde{s}}$
\begin{equation}\label{fumu}
\tilde{\tilde{s}}(x)=\left\{
\begin{array}{cc}
\tilde{s}_{\varepsilon}(x), & x\in ( -\varepsilon ,\varepsilon ), \\
S(x), & x\notin (-\varepsilon ,\varepsilon ).%
\end{array}%
\right.
\end{equation}

Among the simplest functions $\tilde{s}_{\varepsilon }$, which have the advantage to be directly evaluated by computers, are the cubic polynomials $\widetilde{s}_{\varepsilon }:\mathbb{R}\rightarrow \mathbb{R}$ (higher order does not always improve accuracy) defined by
\begin{equation}\label{cubic}
\widetilde{s}_{\varepsilon}(x)=ax^3+bx^2+cx+d,~~ a,b,c,d \in \mathbb{R}.
\end{equation}

Other possible candidates are spline functions, which are constructed via piece-wise polynomials.

\noindent The fact that in (\ref{cond}) there are four coefficients to be determined and only two conditions, means that there are an infinity of choices for $\widetilde{s}_{\varepsilon }$. This implies that $\tilde{s}$, given by (\ref{fumu}), can be even smoothly extended on $\mathbb{R}$, by imposing near the gluing conditions (\ref{cond}), the supplementary differentiability conditions at the boundary of the discontinuity neighborhood
\begin{equation}\label{sm}
\frac{d}{dx}\widetilde{s}_{\varepsilon}(\pm \varepsilon)=\frac{d}{dx}S(\pm \varepsilon).
\end{equation}

For the case of $Sgn$ function, when the gluing and smoothing conditions are $\tilde{s}_{\varepsilon}(\pm\varepsilon)=\pm 1$ and $\frac{d}{d x} \tilde{s}_\varepsilon (\pm \varepsilon)=0$ respectively, the local smooth approximate function, denoted by $\widetilde{sgn}_\varepsilon$, becomes
\begin{equation}\label{loco}
\widetilde{sgn}_{\varepsilon}(x)=-\frac{1}{2\varepsilon^3}x^3+\frac{3}{2\varepsilon}x\approx Sgn(x),~~~  x\in (-\varepsilon, \varepsilon ),
\end{equation}

\noindent and using (\ref{fumu}) on $\mathbb{R}$, $Sgn$ is approximated by the following piece-wise function:
\begin{equation}
\widetilde{\widetilde{sgn}}\approx\left\{\label{loc}
\begin{array}{cc}
\widetilde{sgn}_{\varepsilon}(x), & x\in (-\varepsilon ,\varepsilon ),  \\
\pm1, \text{(or~~} sgn(x)\text{)}, & x\notin [-\varepsilon ,\varepsilon].
\end{array}%
\right.
\end{equation}

\begin{remark}
While $\widetilde{s}$ is not useful for locally approximations (see Fig. \ref{fig2}(a)), $\widetilde{s}_{\varepsilon}$ cannot be used for globally approximation of $S$ since it is unbounded outside the interval $(-\varepsilon ,\varepsilon)$ and tends to $\pm\infty$ as $x\rightarrow\pm\infty$ (see Fig. \ref{fig7}(a)).
\end{remark}

The cubic functions (\ref{cubic}) have a great flexibility, being able to connect smoothly any kind of piece-wise continuous functions on some $\varepsilon$-neighborhood of the discontinuity. For example, by using the smoothness conditions (\ref{sm}), the function (\ref{fifi}) can be smoothly approximated in some neighborhood of the point $x=0.5$, with the cubic function (\ref{cubic}) (see Fig. \ref{fig7}(b) where $\varepsilon$ is chosen to be $1/5$ for a clear image).

Compared to the case of globally approximation, the locally approximations, $\widetilde{s}_{\varepsilon}$, which are determined in the neighborhood of the discontinuity points, at which they are generated, are identical with the single-valued branches of $S$, for $|x|\geq\varepsilon$ (see also Fig. \ref{fig66} for the case of function $f(x)=-10x^2+sgn(x)$).

Using $\widetilde{\widetilde{sgn}}$, for example to Chua's system (\ref{chu}), one obtains
\begin{equation*}
\begin{array}{l}
D_{\ast }^{q_{1}}x_{1}=\left\{
\begin{array}{c}
-2.571x_{1}+9x_{2}+3.857\widetilde{sgn}_{\varepsilon }(x_{1}),~x\in \left( -\varepsilon
,\varepsilon \right),  \\
-2.571x_{1}+9x_{2}+3.857sgn(x_{1}),~x\notin \left( -\varepsilon ,\varepsilon %
\right),
\end{array}%
\right.  \\
D_{\ast }^{q_{2}}x_{2}=x_{1}-x_{2}+x_{3}, \\
D_{\ast }^{q_{3}}x_{3}=-px_{3}.%
\end{array}%
\end{equation*}

\noindent The first component on the right-hand side, which is smoothly approximated, has the image as shown in Fig. \ref{zozo}(b) with, again, a large $\varepsilon$.

The following theorem, similar to Theorem \ref{th1}, states the possibility to approximate locally the right-hand side of IVP (\ref{IVP0}).

\begin{theorem}
Let $f$ be defined by (\ref{f}). If $g$ is continuous, then for every $\varepsilon>0$, there exist local approximations of $f$, $\tilde{f}_\varepsilon:\mathbb{R}^n\rightarrow \mathbb{R}^n$, such that
\begin{equation}\label{loc_th}
\tilde{f}_\varepsilon(x)=g(x)+A(x) {\widetilde{s}_\varepsilon}(x)\approx f(x),~~~ x\in(-\varepsilon,\varepsilon).
\end{equation}
\end{theorem}

\begin{remark}
If $g$ and $\widetilde{s}_\varepsilon$ or $\tilde{s}$ are smooth functions, then one obtains a smooth approximation of $f$. In this case, we can consider to have approximations of class $C^k_\varepsilon(\mathbb{R})$, with $k>1$, and therefore, the IVP (\ref{IVP0}) can be smoothly modeled.
\end{remark}

Summarizing, as can be seen in Fig. \ref{fig55}, aided by Cellina's Theorem, and depending on $g$ properties (continuity or smoothness), the discontinuous function $f$, given by  (\ref{f}), can be continuously or smoothly approximated, by simply replacing the discontinuous function $s$ with either (\ref{h_simplu}) or (\ref{cubic}).

\section{Numerical tests}

In order to illustrate how this approximation apparatus is utilized, we consider two practical examples of piece-wise continuous systems and one theoretical one-dimensional piece-wise continuous system. To emphasize the rightness of the approximation results, the practical examples are of fractional-order.

The use of Caputo derivative in the IVP (\ref{IVP0}) is fully justified in practical examples since in these problems we need physically interpretable initial conditions, i.e., Caputo derivative satisfies these demands. Even there are some applications discussed in recent years with $q > 1$, the great majority of the physical phenomena are modeled with $0<q<1$. Accordingly, the initial condition can be considered in the standard form \cite{kai}, e.g. for the IVP (\ref{IVP0}), $x(0) = x_0$. Therefore, we consider the case of $q<1$.

Two of the most known methods to solve fractional-order equations are the multi-step predictor-corrector Adams-Bashforth-Moulton method (see e.g. \cite{kai,kai2}) and the Gr\"{u}nwald-Letnikov discretization method (see e.g. \cite{unu,doi}). In this paper, we use the Gr\"{u}nwald-Letnikov discretization method with the integration step-size $h=0.005$.

The Hausdorff distance $d_H$ \cite{fal}, used to underline the results rightness, is of order of $10^{-5}$.

\begin{enumerate}
  \item

The fractional variant of the chaotic attractor of piece-wise-linear Chen's system presented in \cite{aziz}, has the following model:
\begin{equation}
\begin{array}{l}\label{ch}
D_{\ast }^{q_1}x_{1}=1.18\left( x_{2}-x_{1}\right),  \\
D_{\ast }^{q_2}x_{2}=sgn(x_{1})\left(5.82-x_{3}\right) +0.7x_{2}, \\
D_{\ast }^{q_3}x_{3}=x_1sgn(x_{2})-0.168x_{3}.
\end{array}
\end{equation}

With the global approximation (\ref{h_simplu}), the system becomes
\begin{equation*}
\begin{array}{l}
D_{\ast }^{q_1}x_{1}=1.18\left( x_{2}-x_{1}\right),  \\
D_{\ast }^{q_2}x_{2}=(5.82-x_3)\widetilde{sgn}(x_{1})+0.7x_{2}, \\
D_{\ast }^{q_3}x_{3}=x_1\widetilde{sgn}(x_{2})-0.1x_{3},
\end{array}%
\end{equation*}

By applying the local approximation (\ref{loc}), one obtains:
\[
\begin{array}{lc}
D_{\ast }^{q_{1}}x_{1}=1.18\left( x_{2}-x_{1}\right) , &  \\
D_{\ast }^{q_{2}}x_{2}=\left\{
\begin{array}{c}
(5.82-x_{3})\widetilde{sgn}_{\varepsilon }(x_{1})+0.7x_{2,} \\
(5.82-x_{3})sgn_{\varepsilon }(x_{1})+0.7x_{2},%
\end{array}%
\right.  &
\begin{array}{c}
x_{1}\in \left( -\varepsilon ,\varepsilon \right) , \\
x_{1}\notin \left( -\varepsilon ,\varepsilon \right) ,%
\end{array}
\\
D_{\ast }^{q_{3}}x_{3}=\left\{
\begin{array}{c}
x_{1}\widetilde{sgn}_{\varepsilon }(x_{2})-0.168x_{3,} \\
x_{1}sgn_{\varepsilon }(x_{2})-0.168x_{3},%
\end{array}%
\right.  &
\begin{array}{c}
x_{2}\in \left( -\varepsilon ,\varepsilon \right) , \\
x_{2}\notin \left( -\varepsilon ,\varepsilon \right) .%
\end{array}%
\end{array}%
\]

\noindent With $\varepsilon=10^{-5}$, and $\delta=10^{-5}$, one obtain the phase plots in Fig. \ref{fig101}, where both (overplotted) attractors have been generated starting from the same initial conditions. As can be seen, both attractors match very well.

~~~~~In the next example, we study a regular motion.

  \item Let us consider a planar mechanical system, an ''inverted`` Duffing-like system (due to the negativeness of the $\dot{x}$ coefficient) of fractional-order modeled by the following equation (see \cite{dan2} for a general form of integer order):
\begin{equation}\label{duf}
\ddot{x}+0.18\dot{x}-1.5x+0.8x^{3}+6.5sgn(\overset{.}{x}
)=35\cos (0.88t).
\end{equation}

The system evolves along a stable limit cycle. In Fig. \ref{fig12} (a), both attractors, determined with local and global approximations, are plotted superimposed, after transients being neglected. Fig. \ref{fig12} (b) and (c) reveal the fact that, for $\delta=10^{-5}$ and $\varepsilon=10^{-5}$, both approximations are of same order of approximation, the difference being of order of $10^{-15}$.

\item Finally, we consider the following one-dimensional piece-wise system of integer-order, which can be found e.g. in \cite{mares} or \cite{don} and which allows us to calculate, empirically, the approximations errors:
\begin{equation}\label{ex}
\dot{x}(t)=2(h(t)-x(t))+h'(t)+2-2sgn(x(t)), ~~~t\in \left[0,2\right],
\end{equation}

\noindent where
\begin{equation}
h(t)=-\frac{4}{\pi}arctan(t-1),~~~x(0)>0.
\end{equation}

As shown in \cite{mares}, the problem (\ref{ex}) admits a unique solution, given by:
\begin{equation}\label{exact}
x(t)=\left\{
\begin{array}{ll}
h(t), & t\in \lbrack 0,1], \\
0, & t\in (1,2].%
\end{array}%
\right.
\end{equation}

The numerical solutions of this system present the sliding phenomena (oscillations), which appear near the discontinuity (manifold $x=0$), and are in full accordance with the convergence order of numerical methods for discontinuous problems (see e.g. \cite{mares} and \cite{don}). In Fig. \ref{ioi} (a), the exact solution and the numerical solutions corresponding to the global and local approximations are plotted respectively. Fig. \ref{ioi} (b) and (c) show the difference between the two approximate trajectories and the exact solution, which is of order of $10^{-5}$, while in Fig. \ref{ioi} (c) this difference is plotted for $t\in[1,2]$. The qualitative difference between these two plots is due to the mentioned sliding phenomenon. However the error is the same in both intervals, namely, of order $10^{-5}$, which are even better than the errors obtained when the IVP is integrated by methods for DE with discontinuous right-hand side \cite{mares}.

The results for this example have been obtained with the Standard Runge-Kutta method, with $h=10^{-5}$ (better errors can be obtained for smaller $h$, but needs a longer computer time).

\end{enumerate}

\section{Conclusion and discussions}

In this paper, we have proven that piece-wise continuous functions defined by (\ref{f}) can be continuously or smoothly approximated. Accordingly, the underlying systems (\ref{IVP0}), of fractional or integer order, can be modeled by continuous or smooth dynamical systems.

The approximations of the discontinuous components can be made locally or globally. This is possible due to the upper semicontinuity of the convex set-valued functions, obtained with Filippov's regularization applied to the piece-wise constant functions $s$, a property which allows the use of the Approximate Cellina's Theorem.

For global approximations, one of the most accessible functions is the sigmoid function (\ref{h_simplu}), while for local approximations, polynomials seem to be the most appropriate choice. However, the steps to prove the existence of these approximations apply for other continuous approximations.

Even the global (polynomial) approximations, from a theoretical point of view, give better performances than the global approximations (due to the more realistic approximation of the underlying physical phenomenon), some aspects related to numerical implementations, require more further studies. Thus, as shown in the last example in the last section, the errors outlineed that both methods give the same accuracy. Therefore, the choice of any one of these approximations should take into account the physical properties of the considered systems. Further studies are to be made.

Although Cellina's Theorem assures as small as desired approximations errors, in the numerical examples we are limited by several kinds of errors, such as the convergence errors of the utilized numerical methods, errors arising from the finite precision representation
of real numbers on computers etc.

The approximation errors in the case of the discontinuous equation (\ref{ex}), where one knows the exact solution, are consistent with the errors of the numerical schemes for discontinuous systems (see e.g. \cite{mares}).

Since there are no numerical methods for fractional piece-wise continuous systems, the approximation apparatus we provided in this paper could be of a real interest. For example, procedures unattainable to piece-wise continuous systems of fractional-order, like chaos control, synchronization, anticontrol, etc., can be attained in this way. Also, these approximation procedures can be utilized for piece-wise continuous systems of integer order modeled by (\ref{IVP0}), so as to obtain continuous systems.

\vspace{3mm}

\textbf{Acknowledgments} We thank Michal Fe\v{c}kan  and Kai Diethlem for helpful suggestions and discussions.

%

\newpage

\begin{figure*}
\begin{center}
  \includegraphics[clip,width=0.4\textwidth] {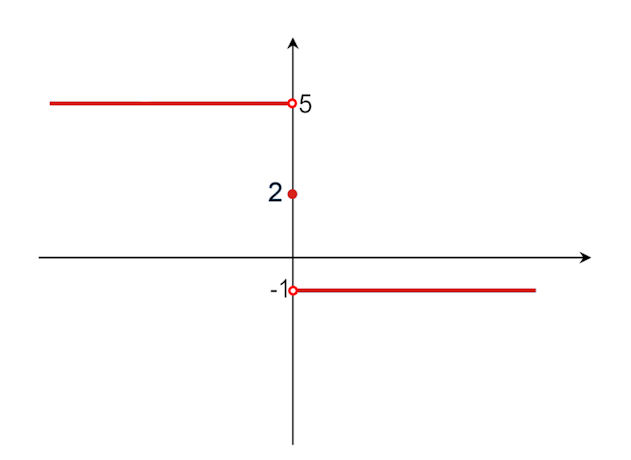}
\caption{Graph of $f(x)=2-3sgn(x)$.}
\label{fig0}
\end{center}
\end{figure*}

\begin{figure*}
\begin{center}
  \includegraphics[clip,width=0.8\textwidth] {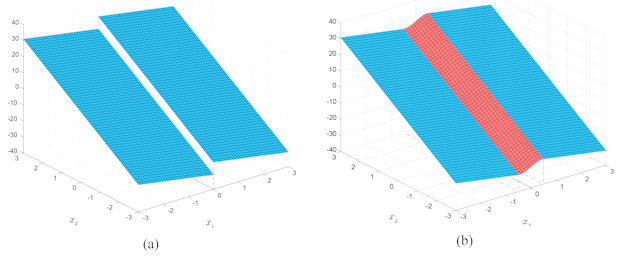}
\caption{Graph of the piece-wise continuous componet of Chua's system (\ref{chu}). (a) Before approximation. (b) After approximation.}
\label{zozo}
\end{center}
\end{figure*}

\begin{figure*}
\begin{center}
  \includegraphics[clip,width=0.6\textwidth] {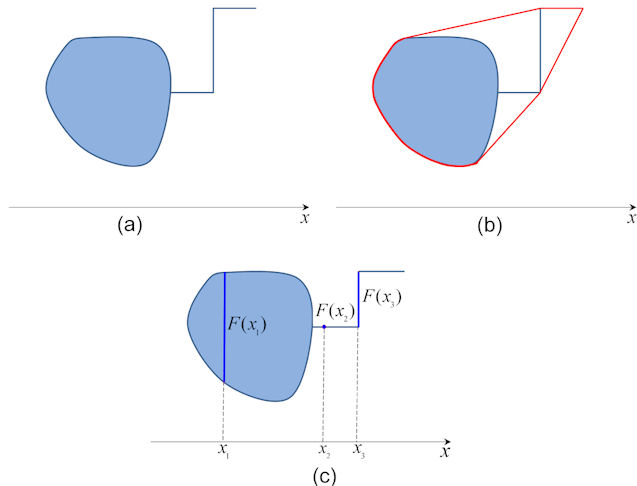}
\caption{(a) Graph of a set-valued function $F$. (b) The closure of the convex hull of $F$. (c) For $x=x_1$ and $x=x_3$, $F(x)$ are segments, while for $x=x_2$, $F(x_2)$ is a point, $f(x_2)$.}
\label{fig00}
\end{center}
\end{figure*}

\begin{figure*}
\begin{center}
  \includegraphics[clip,width=0.5\textwidth] {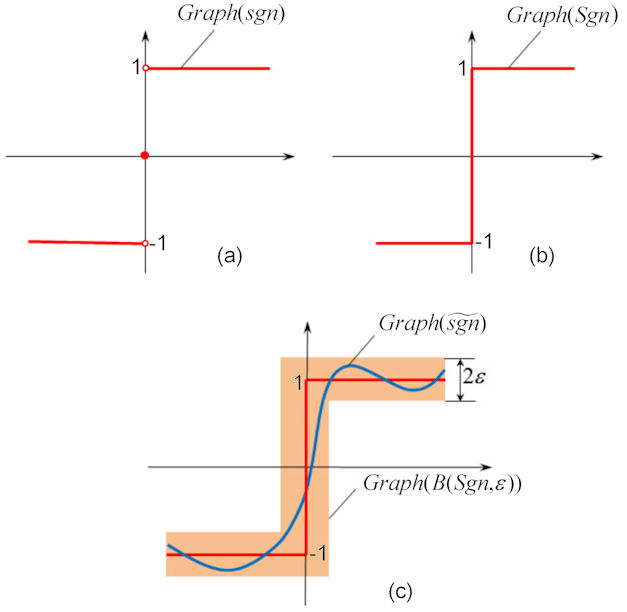}
\caption{(a) Graph of $sgn$. (b) Graph of $Sgn$. (c) Graph of a continuous approximation.}
\label{fig1}
\end{center}
\end{figure*}

\begin{figure*}
\begin{center}
  \includegraphics[clip,width=0.75\textwidth] {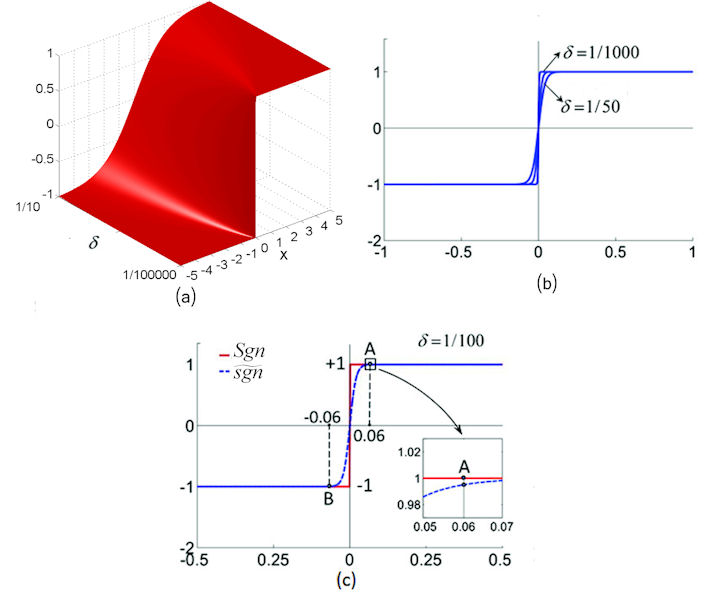}
\caption{(a) Sigmoid function, $\widetilde{sgn}$, for $\delta\in[10^{-5},10^{-1}]$. (b) $\widetilde{sgn}$ for $\delta=1/50$ and $\delta=1/1000$. (c) Graphs of $Sgn$ (red) and of $\widetilde{sgn}$ (blue) for $\delta=1/100$. The detail shows the difference between the two graphs. }
\label{fig2}
\end{center}
\end{figure*}

\begin{figure*}
\begin{center}
  \includegraphics[clip,width=0.5\textwidth] {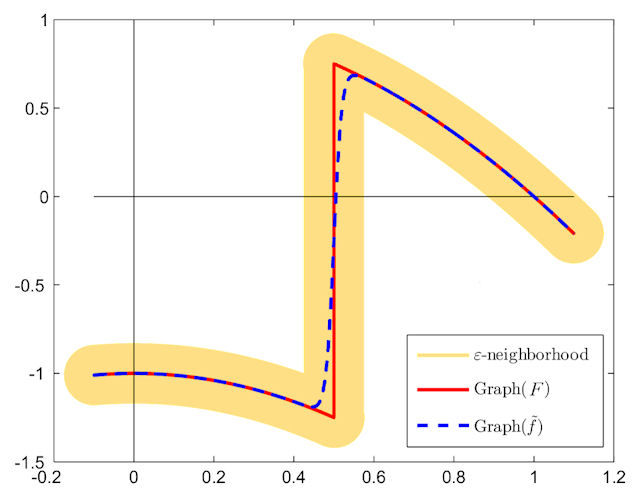}
\caption{Global approximation $\tilde{f}$ (dotted blue color) of the function $f(x)=-x^2+sgn(x-0.5)$. The set-valued form, $F(x)=-x^2+Sgn(x-0.5)$, is plotted in red, and the $\varepsilon$-neighborhood (slowly enlarged) where the approximation, $\tilde{f}$, is embedded, is plotted in yellow.}
\label{fig3}
\end{center}
\end{figure*}

\begin{figure*}
\begin{center}
  \includegraphics[clip,width=0.45\textwidth] {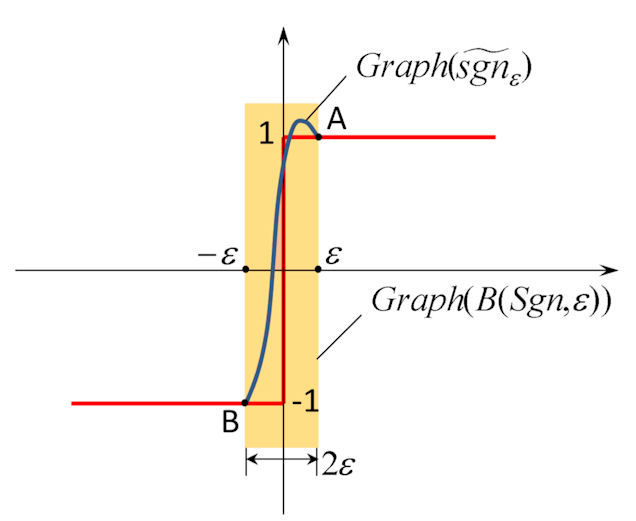}
\caption{Graph of $Sgn$ (red color) and graph of a continuous local approximation $\widetilde{sgn}_\varepsilon$ (blue color), defined inside of some $\varepsilon$-neighborhood of $x=0$. Outside the $\varepsilon$-neighborhood, $\widetilde{sgn}_\varepsilon(x)=Sgn(x).$ }
\label{fig6}
\end{center}
\end{figure*}

\begin{figure*}
\begin{center}
  \includegraphics[clip,width=0.8\textwidth] {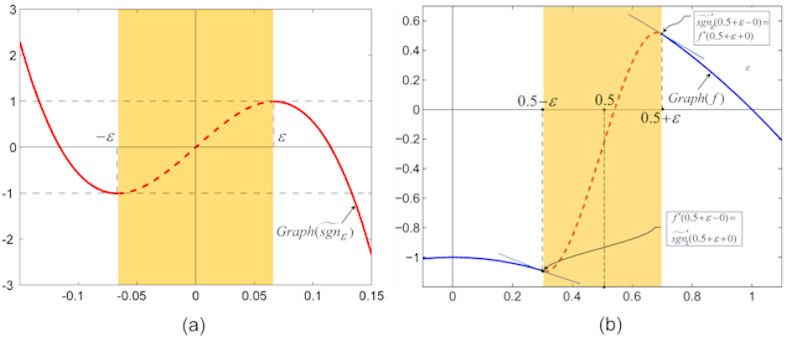}
\caption{(a) Graph of the polynomial local approximation $\widetilde{sgn}_\varepsilon$ of the set-valued $Sgn$ function, defined only within the interval $[-\varepsilon, \varepsilon]$. Outside, $|\widetilde{sgn}_\varepsilon|\rightarrow\infty$. (b) Local approximation of the function $f(x)=-x^2+sgn(x-0.5)$. The approximation is smooth: at the edges of the interval $(-\varepsilon,\varepsilon)$, the supplementary conditions (\ref{sm}) have been imposed (same tangent slope at these points). }
\label{fig7}
\end{center}
\end{figure*}

\begin{figure*}
\begin{center}
  \includegraphics[clip,width=0.5\textwidth] {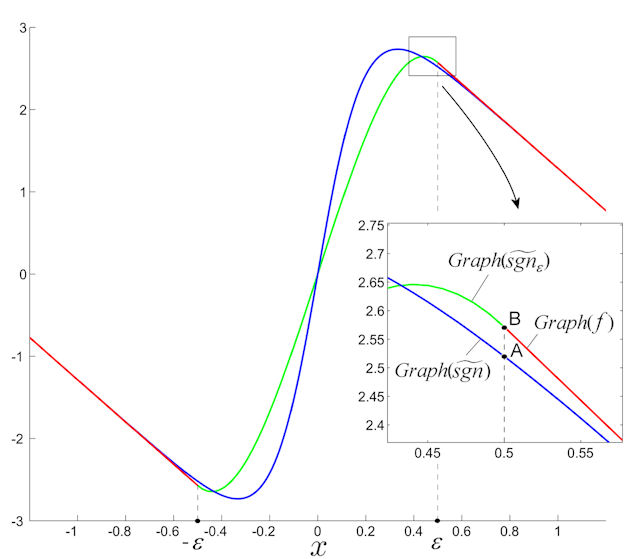}
\caption{Comparison between global and local approximations in the case of piece-wise continuous function $f(x)=-10x^2+sgn(x)$ for a big $\varepsilon=0.5$ value (in order to obtain a clear image). Local approximation $\widetilde{sgn}_\varepsilon$ (green color) is smoothly connected with the graph of $f$ (red color) at the edges point of the $\varepsilon$-interval (point $B$ precisely). The global approximation $\widetilde{sgn}$ is only close to the graph of $f$ (point $A$ precisely).}
\label{fig66}
\end{center}
\end{figure*}

\begin{figure*}
\begin{center}
  \includegraphics[clip,width=0.3\textwidth] {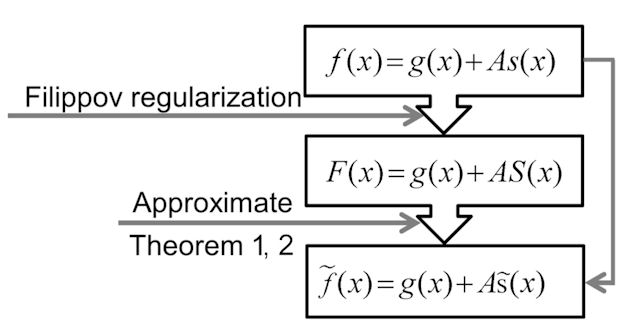}
\caption{Sketch of the proposed approximation procedure. $\tilde{s}$ stands for either local or global approximation.}
\label{fig55}
\end{center}
\end{figure*}

\begin{figure*}
\begin{center}
  \includegraphics[clip,width=0.6\textwidth] {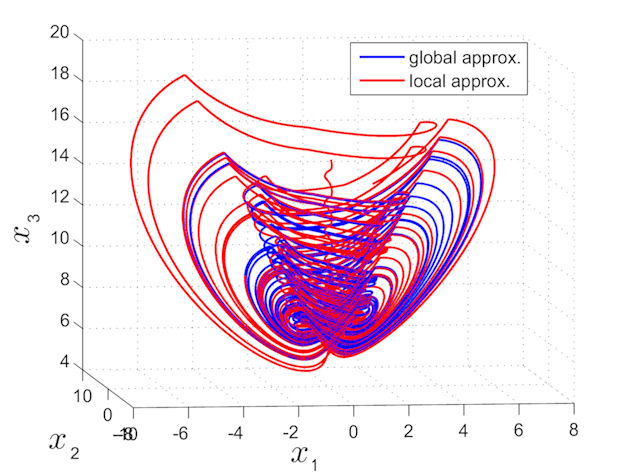}
\caption{Overplotted chaotic attractors of the piece-wise linear Chen system (\ref{ch}), obtained with both approximations.}
\label{fig101}
\end{center}
\end{figure*}

\begin{figure*}
\begin{center}
  \includegraphics[clip,width=0.8\textwidth] {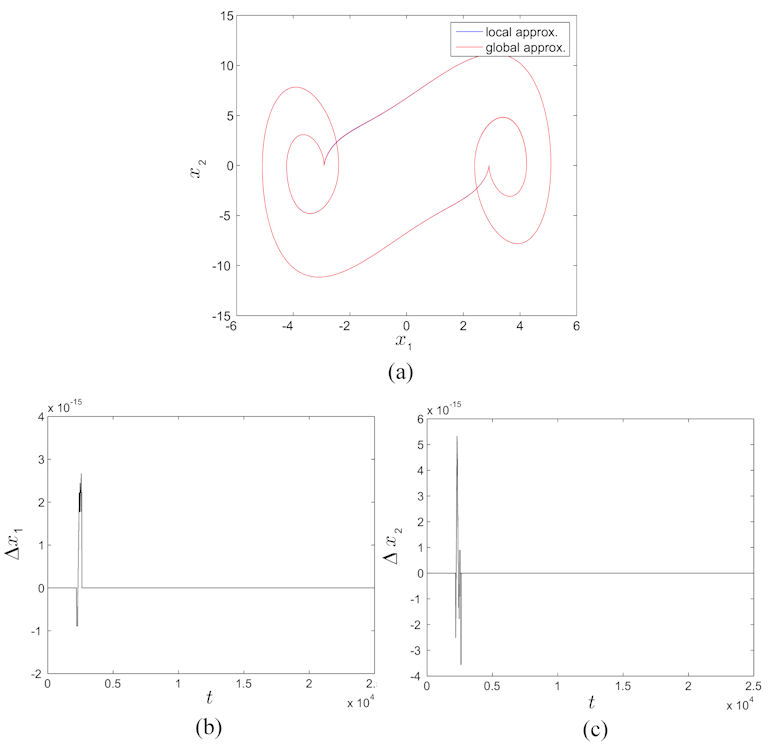}
\caption{(a) Superimposed regular motions of Duffing system (\ref{duf}) corresponding to local (blue color) and global (red color) approximations. (b) Difference between the local and global approximation for the first component $x_1$. (c) Difference between the local and global approximation for the second component $x_2$.}
\label{fig12}
\end{center}
\end{figure*}

\begin{figure*}
\begin{center}
  \includegraphics[clip,width=0.8\textwidth] {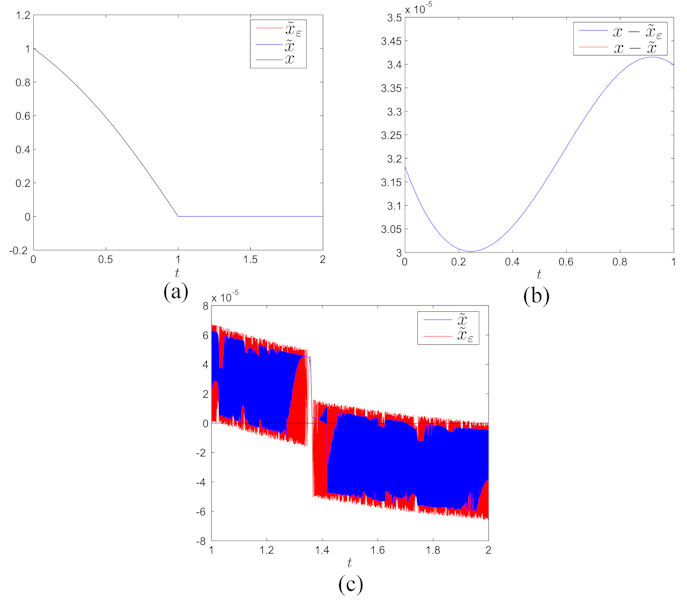}
\caption{(a) Superimposed solutions of the uni-dimensional piece-wise continuous system (\ref{ex}) corresponding to the exact solution (\ref{exact}), $x$ (black), global approximation $\tilde{x}$ (blue) and local approximation $\tilde{x}_\varepsilon$ (red). (b) Difference between the exact solution and the approximate solutions for $t\in[0,1]$. (c) Difference between the exact solution and the approximate solutions for $t\in[1,2]$. The oscillations are typical sliding phenomenon.}
\label{ioi}
\end{center}
\end{figure*}

\end{document}